\documentclass[onecolumn]{article}
\usepackage[utf8x]{inputenc}
\usepackage[dvips]{graphicx}
\usepackage{graphicx}
\usepackage{amsmath}
\usepackage{amsthm}
\usepackage{amsxtra}
\usepackage{array}
\usepackage{threeparttable}
\usepackage{float}

\newfont{\Bb}{msbm10}
\newtheorem{theorem}{Theorem}

\newtheorem{corollary}{\bf Corollary}

\begin{document}

\title{Performance Analysis of the Gradient Comparator LMS Algorithm}

\author{Bijit K. Das, Mrityunjoy Chakraborty}

%
%

\thanks{B. K. Das and M. Chakraborty are with the Department of Electronics and Electrical Communication
Engineering, Indian Institute of Technology, Kharagpur, INDIA
(e.mail : bijitbijit@gmail.com;
 mrityun@ece.iitkgp.ernet.in).}

\abstract{The sparsity-aware zero attractor least mean square (ZA-LMS) algorithm manifests much lower misadjustment in strongly sparse environment
than its sparsity-agnostic counterpart, the least mean square (LMS), but is shown to perform worse than the LMS when sparsity
of the impulse response decreases. The reweighted variant of the ZA-LMS, namely RZA-LMS shows robustness against this variation in sparsity,
but at the price of increased computational complexity. The other variants such as the $l_0$-LMS and the improved proportionate normalized LMS
(IPNLMS), though perform satisfactorily, are also computationally intensive. The gradient comparator LMS (GC-LMS) is a practical solution
of this trade-off when hardware constraint is to be considered. In this paper, we analyse the mean and the mean square convergence performance of
the GC-LMS algorithm in detail. The analyses satisfactorily match with the simulation results.}

\maketitle

\section{Introduction}
Time-varying sparseness is a well-encountered phenomenon in many real-life systems. One of the major examples is echo cancellation in hybrid telephone networks. 
These networks comprising mixed packet-swiched and circuit-switched components require the identification and compensation of echo systems with various level of sparseness. The network echo response in such systems is typically of length
64-128 ms, characterized by a bulk delay dependant on network loading, encoding and jitter buffer delays \cite{1}. This results in an ``active`` region in the range of 8-12 ms duration and consequently, the impulse response is dominated by "inactive'' regions where coefficient magnitudes are close to zero, making the impulse response sparse. The echo canceller must
be ``robust'' to this sparseness \cite{2}. 
  
Acoustic echo, common in hands-free mobile telephony, seriously
degrades user experience due to the coupling between the loudspeaker and microphone. So, effective acoustic
echo cancellation (AEC) \cite{16}. is important to improve the perceived voice quality of a call. The sparsity of
these acoustic impulse responses (AIR) varies with the loudspeaker-microphone distance.
Hence, algorithms developed for mobile hands-free terminals
are required to be robust to the variations in the sparsity of
the acoustic impulse response.
 
      
     Adaptive filters are popular tools for estimating the unknown system parameters. The Least Mean Square (LMS) algorithm, introduced by Widrow
and Hoff \cite{3}, and its variants are popular methods for adaptive system identification. However, standard LMS filters do not exploit sparsity. In the past years, many algorithms exploiting sparsity
were based on 
assigning proportional step sizes of different taps according to their
magnitudes, such as the Proportionate Normalized LMS (PNLMS) \cite{2} and its variants \cite{9}. Motivated by LASSO \cite{10} and recent progress in compressive
sensing \cite{11}-\cite{12},\cite{13} proposed an alternative approach to identifying sparse systems using LMS filters. The basic idea is to introduce
a $l_1$ norm (of the coefficients) penalty which favors sparsity in the cost function. This results in a modified LMS update with a zero attractor for all the taps, naming the Zero-Attracting LMS (ZA-LMS). A variant of the ZA-LMS, namely reweighted ZA-LMS (RZA-LMS) \cite{13},
 shows robustness in case of identifying time-varying sparse impulse response. It employs reweighted step sizes of the zero attractor for the different taps. But, it is associated with huge computational burden due to the $L$ division operations at each step, where $L$ is the number of taps in the filter.

      An alternative method to deal with variable sparseness has been proposed in \cite{14}. That work has been extended in \cite{14b}.
      In this paper, we perform the detailed analysis of the algorithm proposed in \cite{14}.

\section{Review of the ZA-LMS and Corresponding Results \cite{13}}
Taking inspirations from the Least Absolute Shrinkage and Selection Operator (LASSO) \cite{10} and the recent research on Compressive
Sensing (CS) \cite{11}-\cite{12}, a new genre of LMS and RLS variants with $l_0$ or $l_1$
norm constraint is proposed in order to accelerate the sparse system identification. Specifically, by exerting the constraint to the standard
LMS cost function, the solution will be sparse and the gradient
descent recursion will accelerate the convergence of near-zero
coefficients in the sparse system. The next section concentrates on zero-attracting(ZA)-LMS algorithm \cite{13}, the simplest of all compressive sensing based sparsity-aware adaptive algorithms.
In ZA-LMS, a new cost function $L_1(n)$ is defined by combining the instantaneous square error with the $l_1$ norm penalty of the
coefficient vector
\begin{equation}
 L_1(n)=\frac{1}{2}e^2(n)+{\parallel {\bf w}(n) \parallel}_1
\end{equation}
where $e(n)$ is the estimation error of the filter output at the $n^{th}$ time instant, and ${\bf w}(n)$ is the weight vector of the adaptive filter
at the $n^{th}$ instant of time.

$$ e(n) = d(n) - y(n), $$
where
$$ d(n) = {\bf w}_{0}^{T}{\bf x}(n) + e_{0}(n), $$
$$ y(n) = {\bf w}^{T}(n){\bf x}(n), $$
${\bf w}_0$ is the optimum system to be identified,  ${\bf x}(n)$ is the input random vector at the $n^{th}$ instant with autocorrelation matrix ${\bf R},$ 
and $e_{0}(n)$ is the observation noise at the $n^{th}$ instant.

The ZA-LMS filter update equation is 
\begin{equation}
 {\bf w}(n+1)={\bf w}(n)+\mu e(n){\bf x}(n)-\rho sgn\{{\bf w}(n)\}
 \label{eq:zalms}
\end{equation}
where $sgn[\cdotp]$ is a component-wise sign function defined as 
  $$sgn[w_i(n)]=\frac {w_i(n)}{|w_i(n)|}$$ if $w_i(n)]\neq0$; and $0$ otherwise. ($w_i(n)$ is the $i^{th}$ element of the vector ${\bf w}(n)$).

Comparing the ZA-LMS update \eqref{eq:zalms} to the standard LMS update \cite{17}, the ZA-LMS has an additional term $-\rho sgn[{\bf w}(n)]$ which
always attracts the tap coefficients to zero. This is the zero-attractor, whose strength is controlled by $\rho$. Intuitively, the zero-attractor will speed-up convergence when the majority of coefficients
of ${\bf w}_0$ are zero, i.e., the system is sparse. The convergence condition
of the ZA-LMS is provided in the following subsections.\\\\

\subsection{Convergence in mean}
\begin{theorem} [from \cite{13}]
 
In \cite{13}, it has been shown that the mean coefficient vector $E[{\bf w}(n)]$ converges as $n\rightarrow \infty$ if $\mu$ satisfies the condition 
(5) in \cite{13}, and the converged vector is
\begin{equation}
 E[{\bf w}(\infty)]={\bf w}_0-\frac{\rho}{\mu}{\bf R}^{-1}E[sgn\{{\bf w}(\infty)\}]
\end{equation}

\end{theorem}

%
It can be seen that the convergence condition of the ZA-LMS and the standard LMS is same.\\\\

\subsection{Steady-State Excess Mean Square Error [EMSE] }

\begin{theorem}  [from \cite{13}]
It has also been shown in \cite{13} that if NZ denotes the index set of non-zero taps i.e., $w_{0,i}=0$ for $i\in NZ$, and assuming $\rho$ is sufficiently small so that for every $i\in NZ$
\begin{equation}
 E[sgn [w_{i}(\infty)]]=sgn[w_{0,i}],
\end{equation}
 
the excess MSE of the ZA-LMS filter is 
\begin{equation}
 P_{ex}(\infty)=\frac{\eta}{2-\eta}\sigma_v^2+\frac{\alpha_1\rho}{2-\eta}\left(\rho-\frac{2\alpha_2}{\alpha_1}\right),
 \label{eq:ss_emse}
\end{equation}

where
\begin{equation}
 \sigma_v^2=E[e_{0}^{2}(n)] \nonumber
\end{equation}
and,
\begin{equation}
\alpha_1 =  E[sgn({\bf w}(\infty))^T({\bf I} - \mu {\bf R})^{-1}sgn({\bf w}(\infty))]
\end{equation}
with $\bf{I}$ denoting the identity matrix, $\bf{R}$ denoting
the autocovariance matrix of the input, $\sigma_v^2$ indicating the
minimum mean square error, and $\eta\;=\;
Tr(\mu\textbf{R}(\textbf{I}-\mu\textbf{R})^{-1})$, and,
\begin{equation}
\alpha_2 = E[\parallel {\bf w}(\infty)\parallel _1] - \parallel {\bf w}_{0}\parallel _1.
\end{equation}

\end{theorem}

\section{Review of the Gradient Comparator LMS or GC-LMS \cite{14}-\cite{14b}}
\subsection{Comparison of the LMS and the ZA-LMS, and the Motivation behind the GC-LMS}

In \cite{13}, the analysis of convergence behaviour of the ZA-LMS algorithm in mean-square sense has been given. 
The equation (24) in \cite{13} shows how the second term in \eqref{eq:ss_emse} increasing with decreasing system sparsity, makes the ZA-LMS algorithm behave worse 
than even the standard LMS algorithm in less-sparse and non-sparse scenario. Intuitively, one can see the reason behind this. 
The zero-attractor part attracts all the coefficients to zero irrespective of their optimum values. For the actual 'zero' taps, this causes
an accelaration while for the 'non-zero' taps it slowers the convergence. So, when number of 'non-zero' taps increases, the algorithm starts to perform 
worse. In the gradient comparator LMS \cite{14} and its variants \cite{14b}, the zero-attractors are selectively chosen for only the taps having polarity
same as that of the gradient of mean squared instantaneous error. Fig. $1$ and $2$ [originally published in \cite{14} and repeated here for the ease of understanding] show how the GC-LMS manifests better performance than the ZA-LMS
in less sparse systems. Now, for the ease of understanding, we briefly repeat the description of the figures.
In fig. $1$, first we set the values as a highly sparse system with only non-zero element. After $1000$ iterations, the values were reset as
to make a semi-sparse system model. After $2000$ iterations, the values of the elements of the system model vector were again reset as 
to get a highly non-sparse system model. The simulation stoped at $3000^{th}$ iteration. The red curve in the fig. $1$ represented the ZA-LMS, the blue one 
was for the standard LMS, and the green one was for the GC-LMS.
The simulation result in the fig. $1$ showed that in the sparse and semi-sparse system models, the green curve representing the GC-LMS algorithm always became equal to the better among the red (ZA-LMS) and blue (regular LMS) curves. 
In case of highly non-sparse system models, however the performance of GC-LMS was better than ZA-LMS, but worse than LMS.
Fig. $2$ represents the variation of steady-state EMSEs of the LMS (blue), the ZA-LMS (red) and  the GC-LMS (green) algorithms as the functions of system sparsity (which varies from $0$ to $1$)

One can check later that this result is perfectly supported by the mathematical formulation
we will have in the \emph{subsection}$IV.C$.

\subsection{The Compact Form of the GC-LMS Algorithm \cite{14}}

The GC-LMS filter update equation \cite{14}-\cite{14b} is 
\begin{equation}
 { \bf w}(n+1)={ \bf w}(n)+\mu e(n){ \bf x}(n)-\rho {\bf D}(n)  sgn({ \bf w}(n))
 \label{eq:gclms}
\end{equation}
where $sgn[\cdotp]$ is a component-wise sign function defined as 
  $$sgn[w_i(n)]=\frac {w_i(n)}{|w_i(n)|}$$ if $w_i(n)]\neq0$; and $0$ otherwise.
($w_i(n)$ is the $i^{th}$ element of ${\bf w}(n)$ vector)

and, ${\bf D}(n)$ is a diagonal matrix which has $\frac{1}{2}\arrowvert sgn(e(n){ \bf x}(n))-sgn({ \bf w}(n))\arrowvert$ vector as its diagonal.

\begin{figure*}[t]
\begin{center}
\includegraphics[width=150mm]{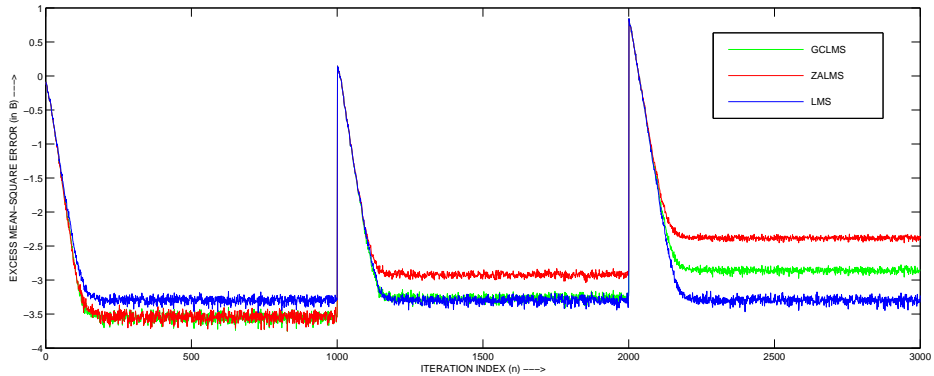}
\end{center}
\caption{The MSE versus no. of observations curve for the standard LMS(blue),
the sparse(ZA)-LMS(red) and the GC-LMS(green) \emph{\textbf{[Published in \cite{14}]}}}
\vspace*{-3pt}
\end{figure*}
%
\begin{figure*}[t]
\begin{center}
\includegraphics[width=150mm]{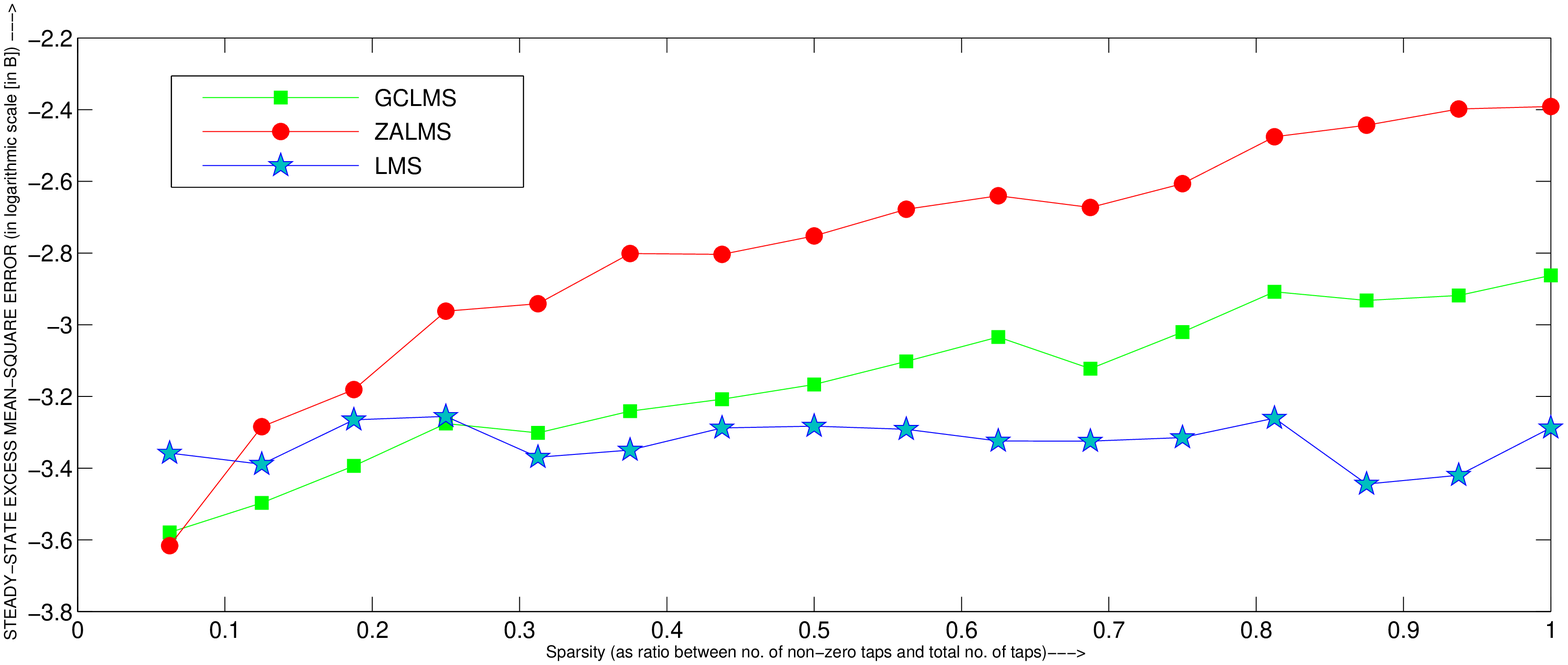}
\end{center}
\caption{The steady-state EMSE versus amount of system sparsity curve for the standard LMS(blue),
the sparse(ZA)-LMS(red) and the GC-LMS(green)  \emph{\textbf{[Published in \cite{14}]}}}
\vspace*{-3pt}
\end{figure*}  


\section{The Performance Analysis of the GC-LMS Algorithm}
In this section, we perform the mean and the mean square convergence analysis of the GC-LMS algorithm. The first part i.e. the \emph{convergence in mean}
has been stated in \cite{14} but no proof was presented in that paper. The second part i.e. the  \emph{mean square convergence} is a wholly new contribution.

\subsection{Convergence in Mean}

\begin{theorem}

The mean coefficient vector $E[{ \bf w}(n)]$ converges as $n\rightarrow \infty$ if $0<\mu<\frac{1}{\lambda_{max}}$ ($\lambda_{max}$ is the maximum eigenvalue of ${\bf R}$, the autocorrelation matrix of the input vector ${\bf x}(n)$), and the converged vector is
\begin{equation}
 E[{ \bf w}(\infty)]={\bf w}_{0}-\frac{\rho}{\mu}{ \bf R}^{-1}E[{\bf D}(\infty)sgn[{ \bf w}(\infty)]]
 \label{eq:gclms0}
\end{equation}

\end{theorem}

\begin{proof} Defining $\widetilde{{ \bf w}}(n)={ \bf w}(n)-{ \bf w}_{0}$, we get from \eqref{eq:gclms}
\begin{eqnarray}
 & & \widetilde{{ \bf w}}(n)=({ \bf I}-\mu { \bf x}(n){ \bf x}^T(n))\widetilde{{ \bf w}}(n-1)-\rho {\bf D}(n)sgn[{ \bf w}(n-1)]\nonumber\\
 & & +\mu e_{0}(n){ \bf x}(n)
 \label{eq:gclms2}
\end{eqnarray}
where $\bf{I}$ denoting the identity matrix, \textbf{R} denoting
the autocovariance matrix of the input,

Taking expectations on both sides of \eqref{eq:gclms2}, there is
\begin{equation}
 E[\widetilde{{ \bf w}}(n)]=({ \bf I}-\mu { \bf R})E[\widetilde{{ \bf w}}(n-1)]-\rho E[{\bf D}(n)sgn[{ \bf w}(n-1)]]
\end{equation}
The vector $\rho {\bf D}(n) sgn[{ \bf w}(n-1)]$ is bounded between $-\rho \bf 1$ and $\rho \bf 1$. Therefore $\rho$, being a very small number,
$E[\widetilde{w}(n)]$ converges if the maximal eigenvalue of $(\bf I-\mu \bf R)$ is less than $1$. Since $E[{ \bf w}(n)]=E[\widetilde{{ \bf w}}(n)]+{ \bf w}_{0}$,
$E[{ \bf w}(n)]$ also converges with the limiting vector shown in \eqref{eq:gclms0}.

It can be seen that the convergence conditions of the GC-LMS, the ZA-LMS and the standard LMS are the same, which are independent with $\rho$. \eqref{eq:gclms0} also implies that the GC-LMS filter returns a biased estimate of the true coefficient vector. 

\end{proof}

\subsection{Mean Square Convergence Analysis and the Steady-State EMSE}

\begin{theorem}
If $NZ$ denotes the index set of non-zero taps i.e., $w_{0,i}=0$ for $i\in NZ$, and assuming $\rho$ is sufficiently small so that for every $i\in NZ$
\begin{equation}
 sgn \{w_i(\infty)\}=sgn\{w_{0,i}\},
\end{equation}
 
the steady-stae EMSE of the GC-LMS filter is 
\begin{equation}
 P_{ex}(\infty)=\frac{ \eta}{2- \eta}\sigma_v^2+\frac{\beta_1\rho}{2-\eta_\mu}\left(\rho-\frac{2\beta_2}{\beta_1}\right),
 \label{eq:ss_gclms}
\end{equation}

where
\begin{equation}
\beta_1\; =\; E[({\bf D}(\infty)sgn\{{ \bf w}(\infty)\})^T({\bf I} - \mu {\bf R})^{-1}{\bf D}(\infty)sgn\{{ \bf w}(\infty)\}],
\end{equation}
 and,
\begin{equation}
\beta_2 = E[\parallel {\bf D}(\infty) { \bf w}(\infty)\parallel _1] - \parallel
{\bf D}(\infty) { \bf w}_{0}\parallel _1.
\end{equation}

\end{theorem}

\begin{proof}

Defining ${ \bf \Phi }(n)=E[({ \bf w}(n)-{ \bf w}_0)({ \bf w}(n)-{ \bf w}_0)^T]$,
we can write the following expression
\begin{eqnarray}
 & & {\bf \Phi }(n) = { \bf \Phi}(n-1)-\mu { \bf R }{ \bf \Phi}(n-1)-\mu { \bf \Phi }(n-1){ \bf R }\nonumber\\
 & &+2 \mu^2 { \bf R }{ \bf \Phi}(n-1){ \bf R }+\mu^2 { \bf R } Tr({ \bf R }{ \bf \Phi }(n-1))\nonumber\\
& & +\mu^2 \sigma_{v}^{2} {\bf R} + \rho^2 E[{\bf D}(n)sgn({ \bf w}(n-1))sgn({ \bf w}^T(n-1)){\bf D}(n)^T]\nonumber \\
& & -\rho E[{\bf D}(n)sgn({ \bf w}(n-1))\widetilde{{ \bf w}}^T(n-1)]({ \bf I}-\mu { \bf R})\nonumber \\
& & -\rho ({ \bf I}-\mu { \bf R})E[{\bf D}(n)sgn({ \bf w}(n-1))\widetilde{{ \bf w}}^T(n-1)]
\end{eqnarray}

Assuming ${\bf x}(n)$ to be white, Gaussian, zero-mean, stationary random process with variance $\sigma_{x}^{2}$,
${\bf R}$ becomes ${\boldsymbol \Lambda}=\sigma_{x}^{2}{\bf I}$

 ${\boldsymbol \Lambda}$ is a diagonal matrix containing eigenvalues of ${ \bf R}$ as diagonal entries.

Then
\begin{eqnarray}
 & & { \bf \Phi}(n)={ \bf \Phi}(n-1)-\mu {\boldsymbol \Lambda}{ \bf \Phi}(n-1)-\mu { \bf \Phi}(n-1){\boldsymbol \Lambda}\nonumber\\
 & & +2 \mu^2 {\boldsymbol \Lambda}{ \bf \Phi}(n-1){\boldsymbol \Lambda}+\mu^2 {\boldsymbol \Lambda} Tr({\boldsymbol \Lambda}{ \bf \Phi}(n-1)) + \mu^2 \sigma_{v}^{2} {\boldsymbol \Lambda}\nonumber\\
& &+\rho^2 E[{\bf D}(n)sgn({ \bf w}(n-1))sgn({ \bf w}^T(n-1)){\bf D}(n)^T]\nonumber \\
& & -\rho E[{\bf D}(n)sgn({ \bf w}(n-1))\widetilde{{ \bf w}}^T(n-1)]({ \bf I}-\mu {\boldsymbol \Lambda})\nonumber \\
& &-\rho ({ \bf I}-\mu {\boldsymbol \Lambda})E[{\bf D}(n)sgn({ \bf w}(n-1))\widetilde{{ \bf w}}^T(n-1)]\nonumber\\
\end{eqnarray}

Now, if we define ${ \boldsymbol \phi}(n)=\{{ \bf \Phi}(n)\}_{i,i},  i=1,2,\dots N$

\begin{equation}
{ \boldsymbol \phi}(n)= { \bf B}{ \boldsymbol \phi}(n-1)+\mu^2  \sigma_{v}^{2} {\boldsymbol \lambda}+\rho^2 { \bf f}(n-1)-2\rho( { \bf I}-\mu {\boldsymbol \Lambda}) { \bf g}(n-1)
\end{equation}
where
$${ \bf B}={ \bf B1}+\mu^2 {\boldsymbol \lambda}{\boldsymbol \lambda}^T$$
${ \bf B1}$ is a diagonal matrix formed by the diagonal entries $$b1_k=1-2\mu \lambda_k+2\mu^2 \lambda_{k}^{2}$$
$${\boldsymbol \lambda}=diag\{{\boldsymbol \Lambda}\}$$
$${ \bf f}(n)=diag\{E[{\bf D}(n)sgn({ \bf w}(n-1))sgn({ \bf w}^T(n-1)){\bf D}(n)^T]\}$$
$${ \bf g}(n)=diag\{E[{\bf D}(n)sgn({ \bf w}(n-1))\widetilde{{ \bf w}}^T(n-1)]\}$$

Now,
$$P_{ex}(\infty)={\boldsymbol \lambda}^T{ \boldsymbol \phi}(\infty)$$

Then,
\begin{equation}
{ \boldsymbol \phi}(\infty)= { \bf B}{ \boldsymbol \phi}(\infty)+\mu^2  \sigma_{v}^{2} {\boldsymbol \lambda}+\rho^2 { \bf f}(\infty)-2\rho( { \bf I}-\mu {\boldsymbol \Lambda}) { \bf g}(\infty)
\end{equation}

The $k^{th}$ entry of ${ \boldsymbol \phi}(\infty)$,
\begin{eqnarray}
 & & \phi_k(\infty)=b1_k\phi_k(\infty)+\mu^2\lambda_kP_{ex}(\infty)+\mu^2\sigma_{v}^{2}\lambda_k\nonumber\\
 & & +\rho^2f_k(\infty)-2\rho(1-\mu\lambda_k)g_k(\infty)
\end{eqnarray}

\begin{eqnarray}
 & &\phi_k(\infty)=\frac{\mu^2\lambda_kP_{ex}(\infty)}{1-b1_k}+\frac{\mu^2\sigma_{v}^{2}\lambda_k}{1-b1_k}+\frac{\rho^2f_k(\infty)}{1-b1_k}\nonumber\\
 & & -\frac{2\rho(1-\mu\lambda_k)g_k(\infty)}{1-b1_k}
\end{eqnarray}

\begin{eqnarray}
& & P_{ex}(\infty)=\sum\limits_{k=1}^{N}\lambda_k\phi_k(\infty)\nonumber=\sum\limits_{k=1}^{N}\frac{\mu^2\lambda_{k}^{2}P_{ex}(\infty)}{1-b1_k}\nonumber\\
& & +\sum\limits_{k=1}^{N}\frac{\mu^2\sigma_{v}^{2}\lambda_{k}^{2}}{1-b1_k}+\sum\limits_{k=1}^{N}\frac{\rho^2\lambda_kf_k(\infty)}{1-b1_k}\nonumber\\
& & -\frac{2\rho\lambda_k(1-\mu\lambda_k)g_k(\infty)}{1-b1_k}
\end{eqnarray}

Now, using $1-b1_k=2\mu\lambda_k-2\mu^2\lambda_k^2$

\begin{eqnarray}
 & & P_{ex}(\infty)=\frac{1}{2}\sum\limits_{k=1}^{N}\frac{\mu\lambda_{k}P_{ex}(\infty)}{1-\mu\lambda}+\frac{1}{2}\sum\limits_{k=1}^{N}\frac{\mu \sigma_{v}^{2}\lambda_{k}}{1-\mu\lambda_k}\nonumber\\
& & +\frac{\rho^2}{2\mu}\sum\limits_{k=1}^{N}\frac{f_k(\infty)}{1-\mu\lambda_k}-\sum\limits_{k=1}^{N}\frac{\rho}{\mu}g_k(\infty)
\end{eqnarray}

\begin{eqnarray}
& &  P_{ex}(\infty)=\frac{tr(\mu{\boldsymbol \Lambda}({\bf I}-\mu{\boldsymbol \Lambda})^{-1})}{2-tr(\mu{\boldsymbol \Lambda}({\bf I}-\mu{\boldsymbol \Lambda})^{-1})}\sigma_{v}^{2}\nonumber\\
& & +\frac{\rho}{ \mu(2-tr(\mu{\boldsymbol \Lambda}({\bf I}-\mu{\boldsymbol \Lambda})^{-1})}(\rho\sum\limits_{k=1}^{N}\frac{f_k(\infty)}{1-\mu\lambda_k}\nonumber\\
& & -\sum\limits_{k=1}^{N}2g_k(\infty))
\end{eqnarray}

For the assumption that the input to be white, the last equation leads to \eqref{eq:ss_gclms}.

\end{proof}

 
%

To further specify $\beta_2$, we have the following corollary.
\begin{corollary}  
If $Z$ and $NZ$ be the index sets of zero taps and non-zero
taps respectively, and ${\bf w}(n)$ is assumed to be Gaussian distributed, an approximation of $\beta_2$ is given by
\begin{equation}
 \beta_2\simeq \sum\limits_{i\in Z}\sqrt{\frac{2}{\pi}{\phi}_{i}(\infty)}-\frac{\rho}{\mu}\sum\limits_{i\in NZ}|s_i(\infty)|,
 \label{eq:beta2}
\end{equation}
where $\phi_{i}(\infty)$ and $s_i$ are $i^{th}$ element of the diagonal of $\boldsymbol{\Phi}(\infty)$ and ${\bf s}$ respectively.

${\bf s}$ is defined as
\begin{eqnarray}
 {\bf s}={\boldsymbol \Lambda}^{-1}E[{\bf D}(\infty)sgn\{{\bf w}(\infty)\}]
\end{eqnarray}
\end{corollary}

\begin{proof}
The proof is exactly similar to the proof of \emph{the Lemma $1$} in \cite{13}. 
\end{proof}

\subsection{Theoretical Comparison with the ZA-LMS}

Now, $\beta_1$ is always positive.

The first term in the R.H.S. of \eqref{eq:ss_gclms} is the excess M.S.E. of the standard LMS based filter.
Therefore when $\beta_2>0$ we can expect lower M.S.E. than standard LMS if 
$\rho$ is selected between $0$ and $\frac{2\beta_2}{\beta_1}$,
$$P_{ex}(\infty)<\frac{\eta}{2-\eta}\sigma_{v}^{2}$$

%


There are two competitive terms in the R.H.S. of \eqref{eq:beta2}. The first one varies about zero for the taps associated with inactive coefficients of ${\bf w}_{0}$. The second term is a bias which is due to the shrinkage of the taps associated with active coefficients of ${\bf w}_{0}$.
When the zero taps take the majority, the first term dominates the second one and positive $\beta_2$ can be therefore obtained. It can be seen that first term in the R.H.S. of \eqref{eq:beta2} is equal to 
that of the equation (24) in \cite{13} related to the ZA-LMS. 
So, for highly sparse system models, when most of the coefficient indices belong to $Z$,  performance of the GC-LMS resembles that of the ZA-LMS.

However, behaviour of the second term in the aforementioned equation is quite different from that of the ZA-LMS in \cite{13} [see eq. (24) in that paper]. 

\begin{corollary}
 When the non-zero taps take majority, $\sum\limits_{i\in NZ}|s_i(\infty)| < \sum\limits_{i\in NZ}|b_i(\infty)|$, where $b_i(n)$ and $s_i(n)$ are $i^{th}$
 elements of ${\bf b}(n)$ and ${\bf s}(n)$ respectively.
\end{corollary}


\begin{proof}

\begin{eqnarray}
& &\sum\limits_{i\in NZ}|s_i(\infty)|\nonumber\\
& & =\frac{\sum\limits_{i\in NZ}E[|sgn\{e(\infty)x_i(\infty)\}-sgn\{ w_{0,i}\}|] |sgn\{ w_{0,i}\}|}{2\sigma_{x}^2}\nonumber\\
\label{eq:s1}
\end{eqnarray}

Now, as we know that the gradient of error surface, i.e., $e^2(n)$ changes its sign periodically at the steady-state, and $sgn\{ w_{0,i}\}$ is a constant, and we can 
also guess that the probability of the gradient to be positive or negative is significantly less than $1$, i.e., $\frac{1}{2}E[|sgn\{e(\infty) x_i(\infty)\}-sgn\{ w_{0,i}\}|]<1$

So, for white, stationary input, \eqref{eq:s1} becomes
\begin{eqnarray}
& & \sum\limits_{i\in NZ}|s_i(\infty)|\nonumber\\
& & < \frac{1}{ \sigma_{x}^{2}}\sum\limits_{i\in NZ}E[ |sgn\{ w_{0,i}\}|]=\sum\limits_{i\in NZ}|b_i(\infty)|\nonumber\\
\label{eq:s2}
\end{eqnarray}

From the above equation \eqref{eq:s2}, \emph{the Corollary $2$} is justified.
\end{proof}

For non-sparse and semi-sparse systems, $\alpha_2$ (for the ZA-LMS) and $\beta_2$ (for the GC-LMS) become negative.
So, both of them perform poorly compared with the LMS. But, using \emph{the Corollary $2$}, we can clearly 
state that the GC-LMS is better than the ZA-LMS for non-sparse and semi-sparse system models.

\section{Conclusions}
Exploiting sparsity in the system model has been widely accepted for
various system identification problems. But only small number of them are
capable of handling the time-varying sparseness, but at the cost of increased computational burden.
The gradient comparator LMS or GC-LMS is a computationally cheap solution to this trade-off. In this paper,
we provide the mean and the mean square convergence analyses of the GC-LMS, and the analytical findings support
the simulation results.

\end{document}